\documentclass[a4paper,UKenglish]{lipics-v2016}
\usepackage{graphicx}
\usepackage{color}

\newcommand{\Oh}[1]{\ensuremath{\mathcal{O}\!\left({#1}\right)}}
\newcommand{\sample}{\ensuremath{{\sf sample}}}
\newcommand{\rank}{\ensuremath{{\sf Rank}}}

\newcommand{\myparagraph}[1]{\vspace*{2mm}\noindent{\bf #1.}}

\newcommand{\Sell}{\ensuremath{S_\ell}}
\renewcommand{\lg}{\log}

\title{An Encoding for Order-Preserving Matching}
\author[1]{Travis Gagie}
\author[2]{Giovanni Manzini}
\author[3]{Rossano Venturini}
\affil[1]{School of Computer Science and Telecommunications, Diego Portales University and CEBIB, Santiago, Chile\\
  \texttt{travis.gagie@mail.udp.cl}}
\affil[2]{Computer Science Institute, University of Eastern Piedmont, Alessandria, Italy and IIT-CNR, Pisa, Italy\\
  \texttt{giovanni.manzini@uniupo.it}}
\affil[3]{Department of Computer Science, University of Pisa, Pisa, Italy and ISTI-CNR, Pisa, Italy\\
  \texttt{rossano.venturini@unipi.it}}
\authorrunning{T. Gagie, G. Manzini and R. Venturini}

\Copyright{Travis Gagie, Giovanni Manzini and Rossano Venturini}


\EventEditors{John Q. Open and Joan R. Acces}
\EventNoEds{2}
\EventLongTitle{42nd Conference on Very Important Topics (CVIT 2016)}
\EventShortTitle{CVIT 2016}
\EventAcronym{CVIT}
\EventYear{2016}
\EventDate{December 24--27, 2016}
\EventLocation{Little Whinging, United Kingdom}
\EventLogo{}
\SeriesVolume{42}
\ArticleNo{23}

\begin{document}

\maketitle

\begin{abstract}
Encoding data structures store enough information to answer the queries they
are meant to support but not enough to recover their underlying datasets.  In
this paper we give the first encoding data structure for the challenging
problem of order-preserving pattern matching. This problem was introduced
only a few years ago but has already attracted significant attention because
of its applications in data analysis. Two strings are said to be an
order-preserving match if the {\em relative order} of their characters is the
same: e.g., \(4, 1, 3, 2\) and \(10, 3, 7, 5\) are an order-preserving match.
We show how, given a string \(S [1..n]\) over an arbitrary alphabet and a
constant \(c \geq 1\), we can build an $\Oh{n \log \log n}$-bit encoding such
that later, given a pattern \(P [1..m]\) with \(m \leq \lg^c n\), we can
return the number of order-preserving occurrences of $P$ in $S$ in $\Oh{m}$
time. Within the same time bound we can also return the starting position of
some order-preserving match for $P$ in $S$ (if such a match exists). We prove
that our space bound is within a constant factor of optimal; our query time
is optimal if $\log \sigma = \Omega(\log n)$. Our space bound contrasts with
the \(\Omega (n \log n)\) bits needed in the worst case to store $S$ itself,
an index for order-preserving pattern matching with no restrictions on the pattern length,
or an index for standard pattern matching even with restrictions
on the pattern length.  Moreover, we can build our encoding knowing only how
each character compares to $\Oh{\lg^c n}$ neighbouring characters.
\end{abstract}

\section{Introduction}
\label{sec:introduction}

As datasets have grown even faster than computer memories, researchers have
designed increasingly space-efficient data structures.  We can now store a
sequence of $n$ numbers from \(\{1, \ldots, \sigma\}\) with \(\sigma \leq n\)
in about $n$ words, and sometimes \(n \lg \sigma\) bits, and sometimes even
\(n H\) bits, where $H$ is the empirical entropy of the sequence, and still
support many powerful queries quickly.  If we are interested only in queries
of the form ``what is the position of the smallest number between the $i$th
and $j$th?'', however, we can do even better: regardless of $\sigma$ or $H$,
we need store only \(2 n + o (n)\) bits to be able to answer in constant
time~\cite{fischer2007new}.  Such a data structure, that stores enough
information to answer the queries it is meant to support but not enough to
recover the underlying dataset, is called an {\em
encoding}~\cite{raman2015encoding}.  As well as the variant of range-minimum
queries mentioned above, there are now efficient encoding data structures for
range top-$k$~\cite{DNRS14,GN15,GINRS13}, range selection~\cite{NRS14}, range
majority~\cite{NT14}, range maximum-segment-sum~\cite{GN15b} and range
nearest-larger-value~\cite{Fis11} on sequences of numbers, and
range-minimum~\cite{GIKRR11} and range
nearest-larger-value~\cite{jayapaul2016space,jo2015compact} on
two-dimensional arrays of numbers; all of these queries return positions but
not values from the sequence or array.
Perhaps Orlandi and Venturini's~\cite{orlandi2016space} results about sublinear-sized data structures for
substring occurrence estimation are the closest to the ones we present in this paper, in that they are more related to pattern matching than range queries: they showed how we can store a
sequence of $n$ numbers from \(\{1, \ldots, \sigma\}\) in significantly less
than \(n \lg \sigma\) bits but such that we can estimate quickly and well how often any
pattern occurs in the sequence.

Encoding data structures can offer better space bounds than traditional data
structures that store the underlying dataset somehow (even in succinct or
compressed form), and possibly even security guarantees: if we can build an
encoding data structure using only public information, then we need not worry
about it being reverse-engineered to reveal private information. From the
theoretical point of view, encoding data structures pose new interesting
combinatorial problems and promise to be a challenging field for future
research.

In this paper we give the first encoding for {\em order-preserving pattern
matching}, which asks us to search in a text for substrings whose characters
have the same relative order as those in a pattern.  For example, in \(6, 3,
9, 2, 7, 5, 4, 8, 1\), the order-preserving matches of \(2, 1, 3\) are \(6,
3, 9\) and \(5, 4, 8\).  Kubica et al.~\cite{ipl/KubicaKRRW13} and Kim et
al.~\cite{tcs/KimEFHIPPT14} formally introduced this problem and gave
efficient online algorithms for it.  Other researchers have continued their
investigation, and we briefly survey their results in
Section~\ref{sec:previous}. As well as its theoretical interest, this problem
has practical applications in data analysis. For example, mining for
correlations in large datasets is complicated by amplification or damping ---
e.g., the euro fluctuating against the dollar may cause the pound to
fluctuate similarly a few days later, but to a greater or lesser extent ---
and if we search only for sequences of values that rise or fall by exactly
the same amount at each step we are likely to miss many potentially
interesting leads. In such settings, searching for sequences in which only
the relative order of the values is constrained to be the same is certainly
more robust.

In Section~\ref{sec:previous} we review some previous work on
order-preserving pattern matching.  In Section~\ref{sec:background} we review
the algorithmic tools we use in the rest of the paper. In
Section~\ref{sec:scanning} we prove our first result showing how, given a
string \(S [1..n]\) over an arbitrary alphabet $[\sigma]$ and a constant \(c
\geq 1\), we can store $\Oh{n \lg \lg n}$ bits --- regardless of $\sigma$ ---
such that later, given a pattern \(P [1..m]\) with \(m < \lg^c n\), in $\Oh{n
\log^c n}$ time we can scan our encoding and report all the order-preserving
matches of $P$ in $S$. Our space bound contrasts with the \(\Omega (n \log
n)\) bits needed in the worst case, when \(\log \sigma = \Omega (\log n)\), to store $S$ itself, an index for order-preserving pattern matching with no restriction on the pattern length, or an index for
standard pattern matching even with restrictions on the pattern length.
(If $S$ is a permutation then we can recover it from an index for unrestricted order-preserving pattern matching,
or from an index for standard matching of patterns of length 2, even when they
do not report the positions of the matches.  Notice this does not contradict Orlandi and Venturini's
result, mentioned above, about estimating substring frequency, since that
permits additive error.)  In fact, we build our representation of $S$ knowing only how each character
compares to \(2 \lg^c n\) neighbouring characters. We show in
Section~\ref{sec:indexing} how to adapt and build on this representation to
obtain indexed order-preserving pattern matching, instead of scan-based,
allowing queries in $\Oh{m \log^3 n}$ time but now reporting the position of
only one match.

In Section~\ref{sec:faster} we give our main result showing how to speed up
our index using weak prefix search and other algorithmic improvements. The
final index is able to {\em count} the number of occurrences and {\em return
the position} of an order-preserving match (if one exists) in
$\Oh{m}$ time. This query time is optimal if $\log \sigma = \Omega(\log n)$.
Finally, in Section~\ref{sec:lowerbound} we show that our space bound is
optimal (up to constant factors) even for data structures that only return
whether or not $S$ contains any order-preserving matches.

\section{Previous Work}
\label{sec:previous}

Although recently introduced, order-preserving pattern matching has received
considerable attention and  has been studied in different settings. For the
online problem, where the pattern is given in advance, the first
contributions were inspired by the classical Knuth-Morris-Pratt and
Boyer-Moore
algorithms~\cite{isaac/BelazzouguiPRV13,ipl/ChoNPS15,tcs/KimEFHIPPT14,ipl/KubicaKRRW13}.
The proposed algorithms have guaranteed linear time worst-case complexity or
sublinear time average complexity. However, for the online problem the best
results in practice are obtained by algorithms based on the concept of
filtration, in which some sort of ``order-preserving'' fingerprint is applied
to the text and the
pattern~\cite{stringology/CantoneFK15,spe/TSOT16,spire/ChhabraGT15,wea/ChhabraT14,ipl/ChhabraT16,FaroK16,DCCpaper}.
This approach was successfully applied also to the harder problem of matching
with errors~\cite{spire/ChhabraGT15,cpm/GawrychowskiU14,wea/HirvolaT14}.

There has also been work on indexed order-preserving pattern matching.
Crochemore et al.~\cite{tcs/CrochemoreIKKLP16} showed how, given a string \(S
[1..n]\), in $\Oh{n \log (n) / \log \log n}$ time we can build an $\Oh{n \log
n}$-bit index such that later, given a pattern \(P [1..m]\), we can return
the starting positions of all the $\mathsf{occ}$ order-preserving matches of
$P$ in $S$ in optimal $\Oh{m + \mathsf{occ}}$ time.  Their index is a kind of
suffix tree, and other researchers~\cite{Sha16personal} are trying to reduce
the space bound to \(n \lg \sigma + o (n \log \sigma)\) bits, where $\sigma$
is the size of the alphabet of $S$, by using a kind of Burrow-Wheeler
Transform instead (similar to~\cite{ganguly2017pbwt}).
Even if they succeed, however, when \(\sigma = n^{\Omega
(1)}\) the resulting index will still take linear space --- i.e., \(\Omega
(n)\) words or \(\Omega (n \log n)\) bits.

In addition to Crochemore et al.'s result, other offline solutions have been
proposed combining the idea of fingerprint and indexing. Chhabra et
al.~\cite{stringology/ChhabraKT15} showed how to speed up the search by
building an FM-index~\cite{isci/FerraginaM01} on the binary string expressing
whether in the input text each element is smaller or larger than the next
one. By expanding this approach, Decaroli et al.~\cite{DCCpaper} show how to
build a compressed file format supporting order-preserving matching without
the need of full decompression. Experiments show that this compressed file
format takes roughly the same space as {\sf gzip} and that in most cases the
search is orders of magnitude faster than the sequential scan of the text. We
point out that these approaches, although interesting for the applications,
do not have competitive worst case bounds on the search cost as we get from
Crochemore et al.'s and in this paper.

\section{Background} \label{sec:background}

In this section we collect a set of algorithmic tools that will be used in
our solutions. In the following we report each result together with a brief
description of the solved problem. More details can be obtained by consulting
the corresponding references. All the results hold in the unit cost word-RAM
model, where each memory word has size $w=\Omega(\log n)$ bits, where $n$ is
the input size. In this model arithmetic and boolean operations between
memory words require $\Oh{1}$ time.

\myparagraph{\rank{} queries on binary vector}
In the next solutions we will need to support $\rank$ queries on a binary vector $B[1..n]$.
Given an index $i$, $\rank(i)$ on $B$ returns the number of $1$s in the prefix $B[1..i]$.
We report here a result in \cite{Jacobson}.

\begin{theorem}\label{thn:rank}
 Given a binary vector $B[1..n]$, we can support $\rank$ queries in constant time
 by using $n+o(n)$ bits of space.
\end{theorem}

\myparagraph{Elias-Fano representation}
In the following we will need to encode an increasing sequence of values in almost
optimal space. There are several solutions to this problem, we report here the result
obtained with the, so-called, Elias-Fano representation \cite{Elias74,Fano71}.

\begin{theorem}\label{thn:ef}
An increasing sequence of $n$ values up to $u$ can be represented by using $\log
{u \choose n}+O(n) = n \log\frac{u}{n} + O(n)$ bits, so that we can access any
value of the sequence in constant time.
\end{theorem}

\myparagraph{Minimal perfect hash functions}
In our solution we will make use of {\em Minimal perfect hash function}s (Mphf) \cite{HT01} and
{\em Monotone minimal perfect hash function}s (Mmphf) \cite{BelazzouguiBPV09}.

Given a subset of $\mathit{S} =\{x_1, x_2, \ldots, x_n\}\subseteq \mathit{U}$
of size n, a minimal perfect hash function has to injectively map keys in $\mathit{S}$
to the integers in $[n]$.
Hagerup and Tholey \cite{HT01} show how to build a space/time optimal minimal perfect hash function as stated by the following theorem.

\begin{theorem} \label{thn:mphf}
Given a subset of $\mathit{S} \subseteq \mathit{U}$ of size $n$, there is
a minimal perfect hash function for $S$ that can be evaluated in
constant time and requires $n\log e +o(n)$ bits of space.
\end{theorem}

A monotone minimal perfect hash function is a Mphf $h()$ that preserves the lexicographic ordering,
i.e., for any two strings $x$ and $y$ in the set, $x \leq y$ if and only if $h(x) \leq h(y)$.
Results on Mmphfs focus their attention on dictionaries of binary strings \cite{BelazzouguiBPV09}.
The results can be easily generalized to dictionaries with strings over larger alphabets.
The following theorem reports the obvious generalization of Theorem 3.1 in
\cite{BelazzouguiBPV09} and Theorem 2 in \cite{WeakPrefix}.

\begin{theorem}\label{thn:monotone}
  Given a dictionary of $n$ strings drawn from the alphabet $[\sigma]$, there is
  a monotone minimal perfect hash function $h()$ that occupies $\Oh{n \log (\ell\log \sigma)}$ bits of space,
  where $\ell$ is the average length of the strings in the dictionary.
  Given a string $P[1..m]$, $h(P)$ is computed in $\Oh{1 + m\log \sigma/w}$ time.
\end{theorem}

\myparagraph{Weak prefix search}
The {\em Prefix Search Problem} is a well-known problem in data-structure
design for strings. It asks for the preprocessing of a given set of $n$ strings 
in such a way that, given a query-pattern $P$, (the lexicographic range of) all
the strings in the dictionary which have $P$ as a prefix can be returned
efficiently in time and space.

Belazzougui {\em et al.} \cite{WeakPrefix} introduced the weak variant of the problem that
allows for a one-sided error in the answer.
Indeed, in the {\em Weak Prefix Search Problem} the answer to a query
is required to be correct only in the case that $P$ is a prefix of at least one
string in dictionary; otherwise, the algorithm returns an arbitrary answer.

Due to these relaxed requirements, the data structures solving the problem are
allowed to use space sublinear in the total length of the indexed strings.
Belazzougui {\em et al.} \cite{WeakPrefix} focus their attention
on dictionaries of binary strings, but their results can be easily generalized
to dictionaries with strings over larger alphabets.
The following theorem states the obvious generalization of Theorem 5 in \cite{WeakPrefix}.

\begin{theorem}\label{thm:weak}
Given a dictionary of $n$ strings drawn from the alphabet $[\sigma]$,
there exists a data structure that weak prefix searches
for a pattern $P[1..m]$ in $\Oh{m\log \sigma/w+\log (m \log\sigma)}$ time.
The data structure uses $\Oh{n \log (\ell\log \sigma)}$ bits of space, where $\ell$ is the
average length of the strings in the dictionary.
\end{theorem}

We remark that the space bound in \cite{WeakPrefix} is better than the one
reported above as it is stated in terms of the hollow trie size of the indexed
dictionary. This measure is always within $\Oh{n \log \ell}$ bits but it may be much better
depending on the dictionary. However, the weaker space bound suffices for
the aims of this paper.

\section{An Encoding for Scan-Based Search}
\label{sec:scanning}

As an introduction to our techniques, we show an $\Oh{n\log\log n}$ bit
encoding supporting scan-based order-preserving matching. Given a sequence
$S[1..n]$ we define the rank encoding \(E(S) [1..n] \) as
\[
E(S)[i] = \left\{ \begin{array}{l@{\hspace{4ex}}l}
0.5 & \parbox{.5\textwidth}{if \(S [i]\) is lexicographically smaller than any
\newline
character in \(\{S [1], \ldots, S [i - 1]\}\),}\\[5ex]
j & \parbox{.5\textwidth}{if \(S [i]\) is equal to the lexicographically $j$th
\newline
character in \(\{S [1], \ldots, S [i - 1]\}\),}\\[5ex]
j + 0.5 & \parbox{.5\textwidth}{if \(S [i]\) is larger than the lexicographically $j$th
\newline
character in \(\{S [1], \ldots, S [i - 1]\}\) but smaller than the lexicographically \((j + 1)\)st,}\\[5ex]
|\{S [1], \ldots, S [i - 1]\}| + 0.5 & \parbox{.5\textwidth}{if \(S [i]\) is lexicographically larger than any
\newline
character in \(\{S [1], \ldots, S [i - 1]\}\).}
\end{array} \right.
\]
This is similar to the representations used in previous papers on
order-preserving matching.  We can build \(E (S)\) in $\Oh{n \log n}$ time.
However, we would ideally need $E(S[i..n])$ for $i=1,\ldots,n$, since
$P[1..m]$ has an order-preserving match in $S[i..i+m-1]$ if and only if $E(P)
= E(S[i..i+m-1])$. Assuming  $P$ has polylogarithmic size, we can devise a
more space efficient encoding.

\begin{lemma}
\label{lem:extracting} Given $S[1..n]$ and a constant $c \geq 1$ let $\ell
=\lg^c n$. We can store $\Oh{n \log \log n}$ bits such that later, given $i$
and \(m \leq \ell\), we can compute \(E (S [i..i + m - 1])\) in $\Oh{m}$
time.
\end{lemma}

\begin{proof}
For every position $i$ in $S$ which is multiple of $\ell = \lg^c n$, we store the ranks
of the characters in the window \(S [i..i+2\ell]\). The ranks are values at most
\(2 \ell\), thus they are stored in $\Oh{\log \ell}$ bits each.
We concatenate the ranks of each window in a vector $V$, which has length
$\Oh{n}$ and takes $\Oh{n \log \ell}$ bits.  Every range \(S [i..i + m - 1]\)
of length \(m \leq \ell\) is fully contained in at least one window and in
constant time we can convert $i$ into $i'$ such that \(V [i'..i' + m - 1]\)
contains the ranks of \(S [i], \ldots, S [i + m - 1]\) in that window.

Computing \(E (S [i..i + m - 1])\) na\"ively from these ranks would take
$\Oh{m \log m}$ time. We can speed up this computation by exploiting the fact
that \(S [i..i + m - 1]\) has polylogaritmic length. Indeed, a recent result
\cite{PatrascuT14} introduces a data structure to represent a small dynamic
set $\cal S$ of $\Oh{w^c}$ integers of $w$ bits each supporting, among the
others, insertions and rank queries in $\Oh{1}$ time. Given an integer $x$,
the rank of $x$ is the number of integers in $\cal S$ that are smaller than
or equal to $x$. All operations are supported in constant time for sets of
size $\Oh{w^c}$. This result allows us to compute $E(S [i..i + m - 1])$ in
$\Oh{m}$ time. Indeed, we can use the above data structure to insert $S [i..i
+ m - 1]$'s characters one after the other and compute their ranks in
constant time.
\end{proof}

It follows from Lemma~\ref{lem:extracting} that given $S$ and $c$, we can
store an $\Oh{n \log \log n}$-bit encoding of $S$ such that later, given a
pattern \(P [1..m]\) with \(m \leq \lg^c n\), we can compute \(E (S [i..i  +
m - 1])\) for each position $i$ in turn and compare it to \(E (P)\), and thus
find all the order-preserving matches of $P$ in $\Oh{n m}$ time.  (It is
possible to speed this scan-based algorithm up by avoiding computing each \(E
(S [i..i + m - 1])\) from scratch but, since this is only an intermediate
result, we do not pursue it further here.)  We note that we can construct the
encoding in Lemma~\ref{lem:extracting} knowing only how each character of $S$
compares to $\Oh{\lg^c n}$ neighbouring characters.

\begin{corollary}
\label{cor:scanning} Given \(S [1..n]\) and a constant \(c \geq 1\), we can
store an encoding of $S$ in $\Oh{n \lg \lg n}$ bits such that later, given a
pattern \(P [1..m]\) with \(m \leq \lg^c n\), we can find all the
order-preserving matches of $P$ in $S$ in $\Oh{n m}$ time.
\end{corollary}

We will not use Corollary~\ref{cor:scanning} in the rest of this paper, but we state it as a baseline easily proven from Lemma~\ref{lem:extracting}.

\section{Adding an Index to the Encoding}
\label{sec:indexing}

Suppose we are given \(S [1..n]\) and a constant \(c \geq 1\). We build the
$\Oh{n \log \log n}$-bit encoding of Lemma~\ref{lem:extracting} for \(\ell =
\lg^c n + \lg n\) and call it $\Sell$.  Using $\Sell$ we can compute $E(S')$
for any substring $S'$ of $S$ of length \(|S'| \leq \ell\) in $\Oh{|S'|}$
time. We now show how to complement $\Sell$ with a kind of ``sampled suffix
array'' using $\Oh{n\log\log n}$ more bits, such that we can search for a
pattern \(P [1..m]\) with \(m \leq \lg^c n\) and return the starting position
of an order-preserving match for $P$ in $S$, if there is one. Out first
solution has $\Oh{m \log^3 n}$ query time; we will improve the query time to
$\Oh{m}$ in the next section.


We define the rank-encoded suffix array \(R [1..n]\) of $S$ such that \(R [i]
= j\) if \(E (S [i..n]) \) is the lexicographically $j$th string in \(\{E
(S[1..n]), E (S [2..n]), \ldots, E (S [n])\}\).  Note that $E(S[i..n])$ has
length~$n-i+1$. Figure~\ref{fig:example} shows an example.

Our algorithm consists of a {\em searching phase}  followed by a {\em
verification phase}.
The goal of the searching phase is to identify a range $[l,r]$ in $R$
which contains all the encodings prefixed by \(E (P)\), if any, or
an arbitrary interval if $P$ does not occur.
The verification phase has to check if there is at least an occurrence of $P$
in this interval, and return one position at which $P$ occurs.

\myparagraph{Searching phase}
Similarly to how we can use a normal suffix array and
$S$ to support normal pattern matching, we could use $R$ and $S$ to find all
order-preserving matches for a pattern \(P [1..m]\) in $\Oh{m \log n}$ time
via binary search, i.e., at each step we choose an index $i$, extract \(S [R
[i]..R [i] + m - 1]\), compute its rank encoding and compare it to \(E (P)\),
all in $\Oh{m}$ time. If $m \leq \ell$ we can compute $E(S[R[i]..R [i] + m -
1])$ using $\Sell$ instead of $S$, still in $\Oh{m}$ time, but storing $R$ still takes \(\Omega (n
\log n)\) bits.

Therefore, for our searching phase we sample and store only every
$\sample$-th element of $R$, by position, and every element of $R$ equal $1$ or $n$ or a
multiple of $\sample$, where $\sample= \lfloor \lg n / \log \log n \rfloor$.
This takes $\Oh{n\log\log n}$ bits. Notice we can still find in $\Oh{m \log
n}$ time via binary search in the sampled $R$ an order-preserving match for
any pattern \(P [1..m]\) that has at least $\sample$ order-preserving matches
in $S$.  If $P$ has fewer than $\sample$ order-preserving matches in $S$ but
we happen to have sampled a cell of $R$ pointing to the starting position of
one of those matches, then our binary search still finds it. Otherwise, we find an interval of
length at most \(\sample - 1\) which contains pointers at least to all the
order-preserving matches for $P$ in $S$; on this interval we perform the
verification phase.

\begin{figure}[t]
\centering
\resizebox{\textwidth}{!}
{$\begin{array}{rrclr@{\hspace{4ex}}llllllllllllllllllllllllllllll}
i & R [i] & L [i] & B [i] & D [i] & E (S [R[i]..n])\\
\hline\\[-.5ex]
1 & {\bf 30} & & & & 0.5 \\
2 & 29 & {\bf 2} & {\bf 1.5} & {\bf 4} & 0.5\ 0.5 \\
3 & 22 & {\bf 2} & {\bf 0.5} & {\bf 2} & 0.5\ 0.5\ 0.5\ 0.5\ 1.5\ 5\ 5.5\ 6.5\ 1 \\
4 & {\bf 13} & & & & 0.5\ 0.5\ 0.5\ 1\ 0.5\ 1.5\ 4\ 4.5\ 1\ 4\ 3.5\ 3.5\ 2\ 3\ 6\ 7\ 7.5\ 2 \\
5 & 2 & {\bf 2} & {\bf 0.5} & {\bf 1} & 0.5\ 0.5\ 0.5\ 1.5\ 2.5\ 3.5\ 5.5\ 2.5\ 2\ 5\ 4\ 8\ 4\ 1\ 1\ 0.5\ 1.5\ 6\ 7\ 1\ 6\ 5\ 4\ 2\ 3\ 6\ 7\ 8\ 2 \\
6 & 23 & {\bf 3} & {\bf 3.5} & {\bf 3} & 0.5\ 0.5\ 0.5\ 1.5\ 4.5\ 5.5\ 6.5\ 1 \\
7 & {\bf 8} & & & & 0.5\ 0.5\ 0.5\ 2.5\ 2.5\ 5.5\ 3\ 0.5\ 1\ 0.5\ 1.5\ 6\ 7\ 1\ 6\ 5\ 4\ 2\ 3\ 6\ 7\ 7.5\ 2 \\
8 & {\bf 14} & & & & 0.5\ 0.5\ 1\ 0.5\ 1.5\ 4\ 4.5\ 1\ 4\ 3.5\ 3.5\ 2\ 3\ 6\ 7\ 7.5\ 2 \\
9 & {\bf 20} & & & & 0.5\ 0.5\ 1.5\ 1.5\ 1.5\ 1.5\ 2.5\ 6\ 7\ 7.5\ 2 \\
10 & 3 & {\bf 3} & {\bf 3.5} & {\bf 1} & 0.5\ 0.5\ 1.5\ 2.5\ 3.5\ 5.5\ 2.5\ 2\ 5\ 4\ 7.5\ 4\ 1\ 1\ 0.5\ 1.5\ 6\ 7\ 1\ 6\ 5\ 4\ 2\ 3\ 6\ 7\ 8\ 2 \\
11 & {\bf 16} & & & & 0.5\ 0.5\ 1.5\ 3.5\ 4.5\ 1\ 4\ 3.5\ 3.5\ 2\ 3\ 6\ 7\ 7.5\ 2 \\
12 & {\bf 24} & & & & 0.5\ 0.5\ 1.5\ 3.5\ 4.5\ 5.5\ 1 \\
13 & 11 & {\bf 2} & {\bf 0.5} & {\bf 3} & 0.5\ 0.5\ 2.5\ 1\ 0.5\ 1\ 0.5\ 1.5\ 4\ 5\ 1\ 4\ 3.5\ 3.5\ 2\ 3\ 6\ 7\ 7.5\ 2 \\
14 & 9 & {\bf 3} & {\bf 3.5} & {\bf 3} & 0.5\ 0.5\ 2.5\ 2.5\ 4.5\ 3\ 0.5\ 1\ 0.5\ 1.5\ 6\ 7\ 1\ 6\ 5\ 4\ 2\ 3\ 6\ 7\ 7.5\ 2 \\
15 & 15 & {\bf 2} & {\bf 1.5} & {\bf 1} & 0.5\ 1\ 0.5\ 1.5\ 3.5\ 4.5\ 1\ 4\ 3.5\ 3.5\ 2\ 3\ 6\ 7\ 7.5\ 2 \\
16 & {\bf 28} & & & & 0.5\ 1.5\ 0.5 \\
17 & 7 & {\bf 3} & {\bf 1.5} & {\bf 4} & 0.5\ 1.5\ 0.5\ 0.5\ 3\ 2.5\ 5.5\ 3\ 0.5\ 1\ 0.5\ 1.5\ 6\ 7\ 1\ 6\ 5\ 4\ 2\ 3\ 6\ 7\ 7.5\ 2 \\
18 & 19 & {\bf 3} & {\bf 1.5} & {\bf 5} & 0.5\ 1.5\ 0.5\ 2\ 1.5\ 1.5\ 1.5\ 2.5\ 6\ 7\ 7.5\ 2 \\
19 & {\bf 12} & & & & 0.5\ 1.5\ 1\ 0.5\ 1\ 0.5\ 1.5\ 4\ 4.5\ 1\ 4\ 3.5\ 3.5\ 2\ 3\ 6\ 7\ 7.5\ 2 \\
20 & {\bf 1} & & & & 0.5\ 1.5\ 1.5\ 0.5\ 2\ 2.5\ 3.5\ 5.5\ 2.5\ 2\ 5\ 4\ 8\ 4\ 1\ 1\ 0.5\ 1.5\ 6\ 7\ 1\ 6\ 5\ 4\ 2\ 3\ 6\ 7\ 8\ 2 \\
21 & 21 & {\bf 2} & {\bf 2.5} & {\bf 1} & 0.5\ 1.5\ 1.5\ 1.5\ 1.5\ 2.5\ 6\ 6.5\ 7.5\ 2 \\
22 & 10 & {\bf 2} & {\bf 1.5} & {\bf 2} & 0.5\ 1.5\ 1.5\ 3.5\ 2\ 0.5\ 1\ 0.5\ 1.5\ 5\ 6\ 1\ 5\ 4.5\ 4\ 2\ 3\ 6\ 7\ 7.5\ 2 \\
23 & 27 & {\bf 4} & {\bf 1.5} & {\bf 2} & 0.5\ 1.5\ 2.5\ 0.5 \\
24 & {\bf 6} & & & & 0.5\ 1.5\ 2.5\ 0.5\ 0.5\ 4\ 3\ 5.5\ 3\ 0.5\ 1\ 0.5\ 1.5\ 6\ 7\ 1\ 6\ 5\ 4\ 2\ 3\ 6\ 7\ 7.5\ 2 \\
25 & 18 & {\bf 4} & {\bf 1} & {\bf 3} & 0.5\ 1.5\ 2.5\ 0.5\ 3\ 2.5\ 2.5\ 2\ 2.5\ 6\ 7\ 7.5\ 2 \\
26 & 26 & {\bf 4} & {\bf 0.5} & {\bf 1} & 0.5\ 1.5\ 2.5\ 3.5\ 0.5 \\
27 & 17 & {\bf 2} & {\bf 2.5} & {\bf 3} & 0.5\ 1.5\ 2.5\ 3.5\ 1\ 3\ 2.5\ 2.5\ 2\ 2.5\ 6\ 7\ 7.5\ 2 \\
28 & {\bf 5} & & & & 0.5\ 1.5\ 2.5\ 3.5\ 1.5\ 1\ 4\ 3\ 5.5\ 3\ 0.5\ 1\ 0.5\ 1.5\ 6\ 7\ 1\ 6\ 5\ 4\ 2\ 3\ 6\ 7\ 7.5\ 2 \\
29 & 25 & {\bf 2} & {\bf 2.5} & {\bf 4} & 0.5\ 1.5\ 2.5\ 3.5\ 4.5\ 1 \\
30 & {\bf 4} & & & & 0.5\ 1.5\ 2.5\ 3.5\ 4.5\ 2.5\ 2\ 5\ 4\ 6.5\ 4\ 1\ 1\ 0.5\ 1.5\ 6\ 7\ 1\ 6\ 5\ 4\ 2\ 3\ 6\ 7\ 7.5\ 2
\end{array}$}
\caption{The rank-encoded suffix array \(R [1..30]\) for \(S [1..30] =
3\,9\,7\,2\,3\,5\,6\,8\,4\,3\,6\,5\,9\,5\,2\,2\,0\,1\,5\,6\,0\,5\,4\,3\,1\,2\,5\,6\,7\,1\), with \(L [i]\), \(B [i]\) and \(D [i]\) computed for \(\sample = 4\).  Stored values are shown in boldface.}
\label{fig:example}
\end{figure}

\myparagraph{Verification phase}
The verification phase receives a range $R [l,r]$ (although $R$ is not stored completely) and has to check if that range contains the starting position of an order preserving match for $P$ and, if so, return its position.
This is done by adding auxiliary data structures to the sampled entries of $R$.

Suppose that for each unsampled element \(R [i] = j\) we store the following
data.
\begin{itemize}
\item the smallest number \(L [i]\) (if one exists) such that \(S [j - 1..j + L [i] - 1]\) has at most \(\lg^c n\) order-preserving matches in $S$;
\item the rank \(B [i] = E(S [j - 1..j + L [i] - 1]^{\mathrm{rev}})[L[i]
    + 1] \leq L [i] + 1 / 2\) of \(S [j - 1]\) in \(S [j..j + L [i] -
    1]\), where the superscript rev indicates that the string is
    reversed;
\item the distance \(D [i]\) to the cell of $R$ containing \(j - 1\) from
    the last sampled element $x$ such that \(E (S [x..x + L [i]])\) is
    lexicographically smaller than \(E (S [j - 1..j + L [i] - 1])\).
\end{itemize}
Figure~\ref{fig:example} shows the values in $L$, $B$ and $D$ for our example.

Assume we are given \(P [1..m]\) and $i$ and told that \(S [R [i]..R[i] + m - 1]\)
is an order-preserving match for $P$, but we are not told the value \(R [i] = j\).
If \(R [i]\) is sampled, of course, then we can return $j$ immediately.  If \(L [i]\)
does not exist or is greater than $m$ then $P$ has at least \(\lg^c n \geq \sample\)\
order-preserving matches in $S$, so we can find one in $\Oh{m}$ time: we consider the sampled values from $R$ that precede and follow \(R [i]\) and check with Lemma~\ref{lem:extracting} whether there are order-preserving matches starting at those sampled values.
Otherwise, from \(L [i]\), \(B [i]\) and $P$, we can compute \(E  (S [j - 1..j + L [i] - 1])\)
in $\Oh{m \log m}$ time: we take the length-\(L [i]\) prefix of $P$; if \(B [i]\) is
an integer, we prepend to \(P [1..L [i]]\) a character equal to the lexicographically
\(B [i]\)th character in that prefix; if \(B [i]\) is \(r + 0.5\) for some integer
$r$ with \(1 \leq r < L [i]\), we prepend a character lexicographically between the
lexicographically $r$th and \((r +1)\)st characters in the prefix;
if \(B [i] = 0.5\) or \(B [i] = L [i] + 0.5\), we prepend a character
lexicographically smaller or larger than any in the prefix, respectively.
We can then find in $\Oh{m \log n}$ time the position in $R$ of $x$, the last
sampled element such that \(E (S [x..x + L [i]])\) is lexicographically smaller
than \(E (S [j - 1..j + L [i] - 1])\).  Adding \(D [i]\) to this position gives
us the position $i'$ of \(j - 1\) in $R$.
Repeating this procedure until we reach a sampled cell of $R$ takes $\Oh{m \log^2 n/\log\log n}=\Oh{m \log^2 n}$
time, and we can then compute and return $j$.  As the reader may have noticed,
the procedure is very similar to how we use backward stepping to locate occurrences of
a pattern with an FM-index~\cite{isci/FerraginaM01}, so we refer to it as a backward step at position $i$.

Even if we do not really know whether \(S [R [i]..R[i] + m - 1]\) is an
order-preserving match for $P$, we can still start at the cell \(R [i]\) and
repeatedly apply this procedure: if we do not find a sampled cell after
\(\sample - 1\) repetitions, then \(S [R [i]..R[i] + m - 1]\) is not an
order-preserving match for $P$; if we do, then we add the number of times we
have repeated the procedure to the contents of the sampled cell to obtain the
contents of \(R [i] = j\). Then, using $\Sell$ we compute \(E(S [j..k + m -
1])\) in $\Oh{m}$ time, compare it to \(E (P)\) and, if they are the same,
return $j$. This still takes $\Oh{m \log^2 n}$ time. Therefore, after our searching phase,
if we find an interval $[l,r]$ of length  at most
\(\sample - 1\) which contains pointers to all the
order-preserving matches for $P$ in $S$ (instead of an order-preserving match
directly), then we can check each cell in that interval with this procedure,
in a total of $\Oh{m \log^3 n}$ time.

If \(R [i] = j\) is the starting position of an order-preserving match for a
pattern \(P [1..m]\) with \(m \leq \lg^c n\) that has at most \(\sample\)
order-preserving matches in $S$, then \(L [i] \leq \lg^c n\). Moreover, if
\(R [i'] = j - 1\) then \(L [i'] \leq \lg^c n + 1\) and, more generally, if
\(R [i''] = j - t\) then \(L [i''] \leq \lg^c n + t\). Therefore, we can
repeat the stepping procedure described above and find $j$ without ever
reading a value in $L$ larger than \(\lg^c n + \lg n\) and, since each value
in $B$ is bounded in terms of the corresponding value in $L$, without ever
reading a value in $B$ larger than \(\lg^c n + \lg n + 1 / 2\). It follows
that we can replace any values in $L$ and $B$ greater than \(\lg^c n + \lg n
+ 1 / 2\) by the flag $-1$, indicating that we can stop the procedure when we
read it.  With this modification, each value in $L$ and $B$ takes $\Oh{\log
\log n}$ bits so, since each value in $D$ is less than \(\lg^c n + \lg n\)
and also takes $\Oh{\log \log n}$ bits, $L$, $B$ and $D$ take a total of
$\Oh{n \log \log n}$ bits. Since also the encoding $\Sell$ from
Lemma~\ref{lem:extracting} with \(\ell = \lg^c n + \lg n\) takes $\Oh{n \log
\log n}$ bits, the following intermediate theorem summarizes our results so far.

\begin{theorem}
\label{thm:slow} Given \(S [1..n]\) and a constant \(c \geq 1\), we can store
an encoding of $S$ in $\Oh{n \lg \lg n}$ bits such that later, given a
pattern \(P [1..m]\) with \(m \leq \lg^c n\), in $\Oh{m \log^3 n}$ time we
can return the position of an order-preserving match of $P$ in $S$ (if one
exists).
\end{theorem}

\myparagraph{A complete search example} Suppose we are searching for
order-preserving matches for \(P = 2\,3\,1\,2\) in the string \(S [1..30]\)
shown in Figure~\ref{fig:example}. Binary search on $R$ tells us that
pointers to all the matches are located in $R$ strictly between \(R [16] =
28\) and \(R [19] = 12\), because
\begin{eqnarray*}
\lefteqn{E (S [28..30] = E (6\,7\,1) = 0.5\;1.5\;0.5}\\
& \prec & E (P) = E (2\,3\,1\,2) = 0.5\;1.5\;0.5\;2\\
& \prec & E (S [12..14]) = E (5\,9\,5) = 0.5\;1.5\;1\,;
\end{eqnarray*}
notice \(R [16] = 28\) and \(R [19] = 12\) are stored because 16, 28 and 12
are multiples of \(\sample = 4\).

We first check whether \(R [17]\) points to an order-preserving match for $P$.  That is, we assume (incorrectly) that it does; we take the first \(L [17] = 3\) characters of $P$; and, because \(B [17] = 1.5\), we prepend a character between the lexicographically first and second, say \(1.5\).  This gives us \(1.5\,2\,3\,1\), whose encoding is \(0.5\,1.5\,2.5\,0.5\).  Another binary search on $R$ shows that \(R [20] = 1\) is the last sampled element $x$ such that \(E (S [x..x + 3])\), in this case \(0.5\,1.5\,1.5\,0.5\), is lexicographically smaller than \(0.5\,1.5\,2.5\,0.5\).  Adding \(D [17] = 4\) to 20, we would conclude that \(R [24] = R [17] - 1\) (which happens to be true in this case) and that \(0.5\,1.5\,2.5\,0.5\) is a prefix of \(E (S [R [24]..n])\) (which also happens to be true).  Since \(R [24] = 6\) is sampled, however, we compute \(E (S [7..10]) = 0.5\,1.5\,0.5\,0.5\) and, since it is not the same as $P$'s encoding, we reject our initial assumption that \(R [17]\) points to an order-preserving match for $P$.

We now check whether \(R [18]\) points to an order preserving match for $P$.  That is, we assume (correctly this time) that it does; we take the first \(L [18] = 3\) characters of $P$; and, because \(B [18] = 1.5\), we prepend a character between the lexicographically first and second, say \(1.5\).  This again gives us \(1.5\,3\,2\,1\), whose encoding is \(0.5\,1.5\,2.5\,0.5\).  As before, a binary search on $R$ shows that \(R [20] = 1\) is the last sampled element $x$ such that \(E (S [x..x + 3])\) is lexicographically smaller than \(0.5\,1.5\,2.5\,0.5\).  Adding \(D [18] = 5\) to 20, we conclude (correctly) that \(R [25] = R [18] - 1\) and that \(0.5\,1.5\,2.5\,0.5\) is a prefix of \(E (S [R [25]..n])\)

Repeating this procedure with \(L [25] = 4\), \(B [25] = 1\) and \(D [25] = 3\), we build a string with encoding \(0.5\,1.5\,2.5\,0.5\), say \(2\,3\,4\,1\), and prepend a character equal to the lexicographically first, 1.  This gives us \(1\,2\,3\,4\,1\), whose encoding is \(0.5\,1.5\,2.5\,3.5\,1\).  Another binary search shows that \(R [24] = 6\) is the last sampled element $x$ such that \(E (S [x..x + 4])\) is lexicographically smaller than \(0.5\,1.5\,2.5\,3.5\,1\).  We conclude (again correctly) that \(R [27] = R [18] - 2\) and that \(0.5\,1.5\,2.5\,3.5\,1\) is a prefix of \(E (S [R [27]..n])\).

Finally, repeating this procedure with \(L [27] = 2\), \(B [27] = 2.5\) and \(D [27] = 3\), we build a string with encoding \(0.5\,1.5\), say \(1\,2\), and prepend a character lexicographically greater than any currently in the string, say 3.  This gives us \(3\,1\,2\), whose encoding is \(0.5\,0.5\,1.5\).  A final binary search show that \(R [8] = 14\) is the last sampled element $x$ such that \(E (S [x..x + 2])\) is lexicographically smaller than \(0.5\,0.5\,1.5\).  We conclude (again correctly) that \(R [11] = R [18] - 3\) and that \(0.5\,0.5\,1.5\) is a prefix of \(E (S [R [11]..n])\).  Since \(R [11] = 16\) is sampled, we compute \(E (S [19..22]) = 0.5\,1.5\,0.5\,2\) and, since it matches $P$'s encoding, we indeed report \(S [19..22]\) as an order-preserving match for $P$.

\section{Achieving $\Oh{m}$ query time}
\label{sec:faster} In this section we prove our main result:

\begin{theorem}
\label{thm:fast} Given \(S [1..n]\) and a constant \(c \geq 1\), we can store
an encoding of $S$ in $\Oh{n \lg \lg n}$ bits such that later, given a
pattern \(P [1..m]\) with \(m \leq \lg^c n\), in $\Oh{m}$ time we can return
the position of an order-preserving match of $P$ in $S$ (if one exists). In
$\Oh{m}$ time we can also report the total number of order-preserving
occurrences of $P$ in $S$.
\end{theorem}

Compared to Theorem~\ref{thm:slow}, we improve the query time from $\Oh{m \log^3 n}$ to $\Oh{m}$.
This is achieved by speeding up several steps of the algorithm described in the previous section.

\myparagraph{Speeding up pattern's encoding}
Given a pattern $P[1..m]$, the algorithm has to compute its encoding $E(P[1..m])$.
Doing this na\"{\i}vely as in the previous section would cost $\Oh{m \log m}$ time,
which is, by itself, larger than our target time complexity.
However, since $m$ is polylogarithmic in $n$, we can speed this up as we sped up the computation of the rank-encoding of \(S [i..i + m - 1]\) in the proof of Lemma~\ref{lem:extracting}, and obtain $E(P)$ in $\Oh{m}$ time.
Indeed, we can insert $P$'s characters one after the other in the data structures of \cite{PatrascuT14}
and compute their ranks in constant time.

\myparagraph{Dealing with short patterns}
The approach used by our solution cannot achieve a $o(\sample)$ query time.
This is because we answer a query by performing $\Theta(\sample)$ backward steps regardless of the pattern's length.
This means that for very short patterns, namely $m=o(\sample)=o(\log n/\log \log n)$,
the solution cannot achieve $\Oh{m}$ query time.
However, we can precompute and store the answers of all these short patterns in $o(n)$ bits.
Indeed, the encoding of a pattern of length at most $m=o(\log n/\log \log n)$ is a binary string
of length $o(\log n)$. Thus, there are $o(\sqrt{n})$ possible encodings.
For each of these encodings we explicitly store the number of its occurrence and the position of one of them in $o(n)$ bits.
From now on, thus, we can safely assume that $m=\Omega(\log n/\log \log n)$.

\myparagraph{Speeding up searching phase}
The searching phase of the previous algorithm has two important drawbacks.
First, it costs $\Oh{m\log n}$ time and, thus, it is obviously too expensive
for our target time complexity.
Second, binary searching on the sampled entries in $R$ gives too imprecise results.
Indeed, it finds a range $[l,r]$ of positions in $R$ which may be potential matches
for $P$. However, if the entire range is within two consecutive sampled
positions, we are only guaranteed that all the occurrences of $P$ are in the range
but there may exist positions in the range which do not match $P$.
This uncertainty forces us to explicitly check {\em every} single position in the range
until a match for $P$ is found, if any.
This implies that we have to check $r-l+1=\Oh{\sample}$ positions in the worst case.
Since every check has a cost proportional to $m$, this gives $\omega(m)$ query time.

We use the data structure for weak prefix search of Theorem~\ref{thm:weak}
to index the encodings of all suffixes of the text truncated at length $\ell = \log^c + \log n$.
This way, we can find the range $[l,r]$
of suffixes prefixed by $E(P[1..m])$ in $\Oh{m\log\log n/w + \log (m \log \log n)}=\Oh{m\log\log n/w + \log \log n}$ time with
a data structure of size $\Oh{n\log\log n}$ bits.
This is because $E(P[1..m])$ is drawn from an alphabet of size $\Oh{\log^c n}$,
and both $m$ and $\ell$ are in $\Oh{\log^c n}$.
Apart from its faster query time, this solution has stronger guarantees.
Indeed, if the pattern $P$ has at least one occurrence, the range $[l,r]$
contains all and only the occurrences of $P$.
Instead, if the pattern $P$ does not occur, $[l,r]$ is an arbitrary and meaningless range.
In both cases, just a single check of any position in the range is enough to answer
the order-preserving query. This property gives a $\Oh{\log n/\log\log n}$ factor improvement over the previous solution.

\myparagraph{Speeding up verification phase}
It is clear by the discussion above that the verification phase has to check only
one position in the range $[l,r]$.
If the range contains at least one sampled entry of $R$, we are done.
Otherwise, we have to perform at most $\sample$ backward steps as in the previous
solution.

We now improve the computation of every single backward step. Assume we have
to compute a backward step at $i$, where $R[i]=j$. Before performing the
backward step, we have to compute the encoding $E(S[j-1..j+L[i]-1])$, given
$B[i]$, $L[i]$, and $E(S[j..j+u])$ for some $u\geq L[i]$. This is done as
follows. We first prepend $0.5$ to $E(S[j..j+u])$ and take its prefix of
length $L[i]$. Then, we increase by one every value in the prefix which is
larger than $B[i]$. These operations can be done in $\Oh{1+ L[i]\log \log n /w}=
\Oh{1+m\log \log n /w}$ time by exploiting word parallelism of the RAM model.
Indeed, we can operate on $O(w/\log\log n)$ symbols of the encoding in
parallel.

Now the backward step at $i$ is $i' = k + D[i]$, where $k$ is the only sampled entry in $R$
whose encoding is prefixed by $E(S[j-1..j+L[i]-1])$. Notice that there cannot be
more than one otherwise $S[j-1..j+L[i]-1]$ would occur more than $\sample$ times,
which was excluded in the construction.

Thus, the problem is to compute $k$, given $i$ and $E(S[j-1..j+L[i]-1])$.
It is crucial to observe that $E(S[j-1..j+L[i]-1])$ depends only on $S$ and $L[i]$ and not on the pattern
$P$ we are searching for. Thus, there exists just one valid $E(S[j-1..j+L[i]-1])$
that could be used at query time for a backward step at $i$. Notice that, if the pattern
$P$ does not occur, the encoding that will be used at $i$ may be different, but in this case it is not necessary
to compute a correct backward step.
Consider the set ${\cal E}$ all these, at most $n$, encodings.
The goal is to map each encoding in ${\cal E}$ to the sampled entry in $R$ that it prefixes.
This can be done as follows.
We build a monotone minimal perfect hash function $h()$ on ${\cal E}$ to map each encoding to its lexicographic rank.
Obviously, the encodings that prefix a certain sampled entry $i$ in $R$ form a consecutive range
in the lexicographic ordering. Moreover, none of these ranges overlaps because each encoding
prefixes exactly one sampled entry.
Thus, we can use a binary vector $B$ to mark each of these ranges, so that, given the lexicographic rank of
an encoding, we can infer the sampled entry it prefixes.
The binary vector is obtained by processing the sampled entries in $R$ in lexicographic order and
by writing the size of its range in unary.
It is easy to see that the sampled entry prefixed by $x=E(S[j-1..j+L[i]-1])$ can be
computed as $\rank_1(h(x))$ in constant time.
The data structures that stores $B$ and supports $\rank$ requires $\Oh{n}$ bits (see Theorem \ref{thn:rank}).

The evaluation of $h()$ is the dominant cost, and, thus, a backward step is
computed in $\Oh{1+m\log\log n/w}$ time.
The overall space usage of this solution is $\Oh{n\log\log n}$ bits,
because $B$ has at most $2n$ bits and $h()$ requires $\Oh{n\log\log n}$ bits
by Theorem \ref{thn:monotone}.


Since we perform at most $\sample$ backward steps, it follows that the overall query time is
$\Oh{\sample\times(1+m\log\log n/w} = \Oh{m}$.
The equality follows by observing that $\sample=\Oh{\log n/\log\log n}$, $m = \Omega(\log n/\log\log n)$ and $w=\Omega(\log n)$.

We finally observe that we could use the weak prefix search data structure
instead of $h()$ to compute a backward step. However, this would introduce a
term $\Oh{\log n}$ in the query time, which would be dominant for short
patterns, i.e., $m=o(\log n)$. 

\myparagraph{Query algorithm}
We report here the query algorithm for a pattern $P[1..m]$, with $m=\Omega(\lg n/ \lg \lg n)$.
Recall that for shorter patterns we store all possible answers.

We first compute $E(P[1..m])$ in $\Oh{1+m\log\log n/w}$ time. Then, we perform
a weak prefix search to identify the range $[l,r]$ of encodings that are prefixed by $E(P[1..m])$
in $\Oh{m\log\log n/w + \log \log n}$ time.
If $P$ has at least one occurrence, the search is guaranteed to find the correct range;
otherwise, the range may be arbitrary but the subsequent check will identify the mistake
and report zero occurrences.

In the checking phase, there are only two possible cases.

The first case occurs when $[l,r]$ contains a sampled entry, say $i$, in $R$.
Thus, we can use the encoding from Lemma~\ref{lem:extracting} to compare $E(S[R[i]..R[i]+m])$ and $E(P[1..m])$ in $\Oh{m}$ time.
If they are equal, we report $R[i]$; otherwise, we are guaranteed that there is no occurrence of $P$ in $S$.

The second case is when there is no sampled entry in $[l,r]$. We arbitrarily
select an index $i \in [l,r]$ and we perform a sequence of backward steps starting from $i$.
If $P$ has at least one occurrence, we are guaranteed to find a sampled entry
$e$ in at most $\sample$ backward steps. The overall time of these backward steps is
$\Oh{\sample \times m\log \log n /w}=\Oh{m}$.
If $e$ is not found, we conclude that $P$ has no occurrence.
Otherwise, we explicitly compare $E(S[R[e]+b..R[e]+m+b])$ and $E(P[1..m])$ in $\Oh{m}$ time, where
$b$ is the number of performed backward steps.
We report $R[e]+b$ only in case of a successful comparison.
Note that if $P$ occurs, then the number of its occurrences is $r-l+1$.

\section{Space Lower Bound} \label{sec:lowerbound}
In this section we prove that our solution is space optimal.
This is done by showing a lower bound on the space that any data structure must use
to solve the easier problem of just establishing if a given pattern $P$ has at least
one order-preserving occurrence in $S$.

More precisely, in this section we prove the following theorem.
\begin{theorem}
  Any encoding data structure
  that indexes any $S[1..n]$ over the alphabet $[\sigma]$ with $\log \sigma = \Omega(\log\log n)$
  which, given a pattern $P[1..m]$ with $m=\lg n$, establishes if $P$ has any order-preserving occurrence in $S$
  must use $\Omega(\log\log n)$ bits of space.
\end{theorem}

By contradiction, we assume that there exists a data structure ${\sf D}$ that
uses $o(n \log \log n)$ bits.
We prove that this implies that we can store any string $S[1,n]$ in less than $n \lg \sigma$ bits, which is clearly impossible.

We start by splitting $S$ into $n/m$ blocks of size $m = \lg n$ characters each.
Let $B_i$ denote the $i$th block in this partition.
Observe that if we know both the list $L(B_i)$ of characters that occur in $B_i$ together with their
number of occurrences and $E(B_i)$, we can recover $B_i$. This is because $E(B_i)$ implicitly tells us how to permute the characters
in $L(B_i)$ to obtain $B_i$. Obviously, if we are able to reconstruct each $B_i$, we can reconstruct $S$.
Thus, our goal is to use ${\sf D}$ together with additional data structures to
obtain $E(B_i)$ and $L(B_i)$, for any $B_i$.

We first directly encode $L(B_i)$ for each $i$ by encoding the sorted sequence of characters with Elias-Fano representation. By Theorem \ref{thn:ef}, we know that
this requires $m\lg \frac{\sigma}{m} + \Oh{m}$ bits. Summing up over all the blocks,
the overall space used is $n \lg \frac{\sigma}{m} + \Oh{n}$ bits.

Now it remains to obtain the encodings of all the blocks. Consider the set
${\cal E}$ of the encodings of all the substrings of $S$ of length $m$. We do
not store ${\cal E}$ because it would require too much space. Instead, we use
a minimal perfect hash function $h()$ on ${\cal E}$. This requires $\Oh{n}$
bits by Theorem \ref{thn:mphf}. This way each distinct encoding is
bijectively mapped to a value in $[n]$. For each block $B_i$, we store
$h(B_i)$. This way, we are keeping track of those elements in ${\cal E}$ that
are blocks and their positions in $S$. This requires $\Oh{n}$ bits, because
there are $n/\lg n$ blocks and storing each value needs $\Oh{\log n}$ bits.

We are now ready to retrieve the encoding of all the blocks, which is the last step to
be able to reconstruct $S$.
This is done by searching in $D$ for every possible encoding of exactly $m$ characters.
The data structure will be able to tell us the ones that occurs in $S$, i.e., we are
retrieving the entire set ${\cal E}$.
For each encoding $e \in {\cal E}$, we check if $h(e)$ is the hash of any of the blocks.
In this way we are able to associate the encodings in ${\cal E}$ to the original block.

Thus, we are able to reconstruct $S$ by using ${\sf D}$ and additional data structures which uses
$n \lg \sigma - n\lg\lg n + \Oh{n}$ bits of space.
This implies that ${\sf D}$ cannot use $o(n\log\log n)$ bits.

\section{Conclusion}

We have given an encoding data structure for order-preserving pattern matching:
given a string $S$ of length $n$ over an arbitrary alphabet and a constant
\(c \geq 1\), we can store $\Oh{n \log \log n}$ bits such that later,
given a pattern $P$ of length \(m \leq \lg^c n\), in $\Oh{m}$ time we
can return the position of an order-preserving match of $P$ in $S$ (if one exists)
and report the number of such matches.  Our space bound is within a constant factor
of optimal, even for only detecting whether a match exists,
and our time bound is optimal when the alphabet size is at least logarithmic in $n$.
We can build our encoding knowing only how each character of $S$ compares to
$\Oh{\lg^c n}$ neighbouring characters.
We believe our results will help open up a new line of research,
where space is saved by restricting the set of possible queries
or by relaxing the acceptable answers,
that will help us deal with the rapid growth of datasets.

\newpage

\bibliographystyle{plainurl}
\bibliography{op-encoding}

\begin{thebibliography}{10}

\bibitem{BelazzouguiBPV09}
Djamal Belazzougui, Paolo Boldi, Rasmus Pagh, and Sebastiano Vigna.
\newblock Monotone minimal perfect hashing: searching a sorted table with o (1)
  accesses.
\newblock In {\em Proceedings of the twentieth Annual ACM-SIAM Symposium on
  Discrete Algorithms}, pages 785--794. SIAM, 2009.

\bibitem{WeakPrefix}
Djamal Belazzougui, Paolo Boldi, Rasmus Pagh, and Sebastiano Vigna.
\newblock Fast prefix search in little space, with applications.
\newblock In {\em European Symposium on Algorithms}, pages 427--438. Springer,
  2010.

\bibitem{isaac/BelazzouguiPRV13}
Djamal Belazzougui, Adeline Pierrot, Mathieu Raffinot, and St{\'e}phane
  Vialette.
\newblock Single and multiple consecutive permutation motif search.
\newblock In {\em International Symposium on Algorithms and Computation}, pages
  66--77. Springer, 2013.

\bibitem{stringology/CantoneFK15}
Domenico Cantone, Simone Faro, and M~Oguzhan K{\"u}lekci.
\newblock An efficient skip-search approach to the order-preserving pattern
  matching problem.
\newblock In {\em Stringology}, pages 22--35, 2015.

\bibitem{spe/TSOT16}
Tamanna Chhabra, Simone Faro, M~O{\u{g}}uzhan K{\"u}lekci, and Jorma Tarhio.
\newblock Engineering order-preserving pattern matching with simd parallelism.
\newblock {\em Software: Practice and Experience}, 2016.

\bibitem{spire/ChhabraGT15}
Tamanna Chhabra, Emanuele Giaquinta, and Jorma Tarhio.
\newblock Filtration algorithms for approximate order-preserving matching.
\newblock In {\em International Symposium on String Processing and Information
  Retrieval}, pages 177--187. Springer, 2015.

\bibitem{stringology/ChhabraKT15}
Tamanna Chhabra, M~Oguzhan K{\"u}lekci, and Jorma Tarhio.
\newblock Alternative algorithms for order-preserving matching.
\newblock In {\em Stringology}, pages 36--46, 2015.

\bibitem{wea/ChhabraT14}
Tamanna Chhabra and Jorma Tarhio.
\newblock Order-preserving matching with filtration.
\newblock In {\em International Symposium on Experimental Algorithms}, pages
  307--314. Springer, 2014.

\bibitem{ipl/ChhabraT16}
Tamanna Chhabra and Jorma Tarhio.
\newblock A filtration method for order-preserving matching.
\newblock {\em Information Processing Letters}, 116(2):71--74, 2016.

\bibitem{ipl/ChoNPS15}
Sukhyeun Cho, Joong~Chae Na, Kunsoo Park, and Jeong~Seop Sim.
\newblock A fast algorithm for order-preserving pattern matching.
\newblock {\em Information Processing Letters}, 115(2):397--402, 2015.

\bibitem{tcs/CrochemoreIKKLP16}
Maxime Crochemore, Costas~S Iliopoulos, Tomasz Kociumaka, Marcin Kubica,
  Alessio Langiu, Solon~P Pissis, Jakub Radoszewski, Wojciech Rytter, and
  Tomasz Wale{\'n}.
\newblock Order-preserving indexing.
\newblock {\em Theoretical Computer Science}, 638:122--135, 2016.

\bibitem{DNRS14}
Pooya Davoodi, Gonzalo Navarro, Rajeev Raman, and S~Srinivasa Rao.
\newblock Encoding range minima and range top-2 queries.
\newblock {\em Philosophical Transactions of the Royal Society of London A:
  Mathematical, Physical and Engineering Sciences}, 372(2016):20130131, 2014.

\bibitem{DCCpaper}
Gianni Decaroli, Travis Gagie, and Giovanni Manzini.
\newblock A compact index for order-preserving pattern matching.
\newblock In {\em Data Compression Conference}, 2017.
\newblock To appear.

\bibitem{Elias74}
Peter Elias.
\newblock Efficient storage and retrieval by content and address of static
  files.
\newblock {\em Journal of the ACM (JACM)}, 21(2):246--260, 1974.

\bibitem{Fano71}
Robert~M. Fano.
\newblock On the number of bits required to implement an associative memory.
\newblock Technical Report Memorandum 61, Project MAC, Computer Structures
  Group, Massachusetts Institute of Technology, 1971.

\bibitem{FaroK16}
Simone Faro and M~O{\u{g}}uzhan K{\"u}lekci.
\newblock Efficient algorithms for the order preserving pattern matching
  problem.
\newblock In {\em International Conference on Algorithmic Applications in
  Management}, pages 185--196. Springer, 2016.

\bibitem{isci/FerraginaM01}
Paolo Ferragina and Giovanni Manzini.
\newblock An experimental study of a compressed index.
\newblock {\em Information Sciences}, 135(1):13--28, 2001.

\bibitem{Fis11}
Johannes Fischer.
\newblock Combined data structure for previous-and next-smaller-values.
\newblock {\em Theoretical Computer Science}, 412(22):2451--2456, 2011.

\bibitem{fischer2007new}
Johannes Fischer and Volker Heun.
\newblock A new succinct representation of rmq-information and improvements in
  the enhanced suffix array.
\newblock In {\em Combinatorics, Algorithms, Probabilistic and Experimental
  Methodologies}, pages 459--470. Springer, 2007.

\bibitem{ganguly2017pbwt}
Arnab Ganguly, Rahul Shah, and Sharma~V Thankachan.
\newblock pbwt: Achieving succinct data structures for parameterized pattern
  matching and related problems.
\newblock In {\em Proceedings of the Twenty-Eighth Annual ACM-SIAM Symposium on
  Discrete Algorithms}, pages 397--407. SIAM, 2017.

\bibitem{GN15b}
Pawe{\l} Gawrychowski and Patrick~K Nicholson.
\newblock Encodings of range maximum-sum segment queries and applications.
\newblock In {\em Annual Symposium on Combinatorial Pattern Matching}, pages
  196--206. Springer, 2015.

\bibitem{GN15}
Pawe{\l} Gawrychowski and Patrick~K Nicholson.
\newblock Optimal encodings for range top-k, selection, and min-max.
\newblock In {\em International Colloquium on Automata, Languages, and
  Programming}, pages 593--604. Springer, 2015.

\bibitem{cpm/GawrychowskiU14}
Pawe{\l} Gawrychowski and Przemys{\l}aw Uzna{\'n}ski.
\newblock Order-preserving pattern matching with $k$ mismatches.
\newblock {\em Theoretical Computer Science}, 638:136--144, 2016.

\bibitem{GIKRR11}
Mordecai Golin, John Iacono, Danny Krizanc, Rajeev Raman, Srinivasa~Rao Satti,
  and Sunil Shende.
\newblock Encoding 2d range maximum queries.
\newblock {\em Theoretical Computer Science}, 609:316--327, 2016.

\bibitem{GINRS13}
Roberto Grossi, John Iacono, Gonzalo Navarro, Rajeev Raman, and Satti~Srinivasa
  Rao.
\newblock Encodings for range selection and top-k queries.
\newblock In {\em European Symposium on Algorithms}, pages 553--564. Springer,
  2013.

\bibitem{HT01}
Torben Hagerup and Torsten Tholey.
\newblock Efficient minimal perfect hashing in nearly minimal space.
\newblock In {\em Annual Symposium on Theoretical Aspects of Computer Science},
  pages 317--326. Springer, 2001.

\bibitem{wea/HirvolaT14}
Tommi Hirvola and Jorma Tarhio.
\newblock Approximate online matching of circular strings.
\newblock In {\em International Symposium on Experimental Algorithms}, pages
  315--325. Springer, 2014.

\bibitem{Jacobson}
Guy Jacobson.
\newblock Space-efficient static trees and graphs.
\newblock In {\em Foundations of Computer Science, 1989., 30th Annual Symposium
  on}, pages 549--554. IEEE, 1989.

\bibitem{jayapaul2016space}
Varunkumar Jayapaul, Seungbum Jo, Rajeev Raman, Venkatesh Raman, and
  Srinivasa~Rao Satti.
\newblock Space efficient data structures for nearest larger neighbor.
\newblock {\em Journal of Discrete Algorithms}, 36:63--75, 2016.

\bibitem{jo2015compact}
Seungbum Jo, Rajeev Raman, and Srinivasa~Rao Satti.
\newblock Compact encodings and indexes for the nearest larger neighbor
  problem.
\newblock In {\em International Workshop on Algorithms and Computation}, pages
  53--64. Springer, 2015.

\bibitem{tcs/KimEFHIPPT14}
Jinil Kim, Peter Eades, Rudolf Fleischer, Seok-Hee Hong, Costas~S Iliopoulos,
  Kunsoo Park, Simon~J Puglisi, and Takeshi Tokuyama.
\newblock Order-preserving matching.
\newblock {\em Theoretical Computer Science}, 525:68--79, 2014.

\bibitem{ipl/KubicaKRRW13}
Marcin Kubica, Tomasz Kulczy{\'n}ski, Jakub Radoszewski, Wojciech Rytter, and
  Tomasz Wale{\'n}.
\newblock A linear time algorithm for consecutive permutation pattern matching.
\newblock {\em Information Processing Letters}, 113(12):430--433, 2013.

\bibitem{NRS14}
Gonzalo Navarro, Rajeev Raman, and Srinivasa~Rao Satti.
\newblock Asymptotically optimal encodings for range selection.
\newblock In {\em 34th International Conference on Foundation of Software
  Technology and Theoretical Computer Science}, page 291, 2014.

\bibitem{NT14}
Gonzalo Navarro and Sharma~V Thankachan.
\newblock Encodings for range majority queries.
\newblock In {\em CPM}, pages 262--272, 2014.

\bibitem{orlandi2016space}
Alessio Orlandi and Rossano Venturini.
\newblock Space-efficient substring occurrence estimation.
\newblock {\em Algorithmica}, 74(1):65--90, 2016.

\bibitem{PatrascuT14}
Mihai Patrascu and Mikkel Thorup.
\newblock Dynamic integer sets with optimal rank, select, and predecessor
  search.
\newblock In {\em Foundations of Computer Science (FOCS), 2014 IEEE 55th Annual
  Symposium on}, pages 166--175. IEEE, 2014.

\bibitem{raman2015encoding}
Rajeev Raman.
\newblock Encoding data structures.
\newblock In {\em International Workshop on Algorithms and Computation}, pages
  1--7. Springer, 2015.

\bibitem{Sha16personal}
Rahul Shah.
\newblock Personal communication, 2016.

\end{thebibliography}

\end{document}